\documentclass[oneside,letterpaper]{scrartcl} \usepackage{macros} %

\newcommand{\eps}{\epsilon}
\newcommand{\mcard}{q}%
\newcommand{\nmat}{q}%

\newcommand{\mindex}{\ell} %

\newcommand{\pmatchoid}{\matchoid^p} %
\newcommand{\defpmatchoid}{\pmatchoid = (\groundset,\independents)} %
\newcommand{\matroidI}{\matroid_{\mindex}} %
\newcommand{\groundsetI}{\groundset_{\mindex}} %
\newcommand{\independentsI}{\independents_{\mindex}} %
\newcommand{\defmatroidI}{\matroidI = (\groundsetI,\independentsI)} %
\newcommand{\offlineratio}{\gamma_p} %
\newcommand{\optset}{T^*} %

\begin{document}

\title{Streaming Algorithms for \\Submodular Function Maximization} %
\author{ Chandra Chekuri\thanks{ %
    Work on this paper supported in part by NSF grant
    CCF-1319376.  
  } %
  \and %
  Shalmoli Gupta\thanks{%
    Work on this paper supported in part by NSF grant
    CCF-1319376.  
  }\and %
  Kent Quanrud\thanks{
    Work on this paper supported in part by NSF grants CCF-1319376,
    CCF-1421231, and CCF-1217462.  
  } %
  \and
  \\
  Dept.\ of Computer Science, Univ.\ of Illinois, Urbana IL 61801, USA \\
  \email{\{chekuri,sgupta49,quanrud2\}@illinois.edu} } %

\maketitle

\begin{abstract}
  We consider the problem of maximizing a nonnegative submodular set
  function $f:2^{\groundset} \rightarrow \reals^+$ subject to a
  $p$-matchoid constraint in the single-pass streaming
  setting. Previous work in this context has considered streaming
  algorithms for modular functions and monotone submodular
  functions. The main result is for submodular functions that are {\em
    non-monotone}. We describe deterministic and randomized algorithms
  that obtain a $\Omega(\frac{1}{p})$-approximation using $O(k \log
  k)$-space, where $k$ is an upper bound on the cardinality of the
  desired set. The model assumes value oracle access to $f$ and
  membership oracles for the matroids defining the $p$-matchoid
  constraint.
\end{abstract}

\section{Introduction}

Let $f: \powerset{\groundset} \to \reals$ be a set function defined
over a ground set $\groundset$. $f$ is \newterm{submodular} if it
exhibits decreasing marginal values in the following sense: if $e \in
\groundset$ is any element, and $A,B \subseteq \groundset$ with $A
\subseteq B$ are any two nested sets, then $f(A + e) - f(A) \geq f(B +
e) - f(B)$. The gap $f(A+e) - f(A)$ is called the \newterm{marginal
  value} of $e$ with respect to $f$ and $A$, and denoted $f_A(e)$. An
equivalent characterization for submodular functions is that for any
two sets $A,B \subseteq \groundset$, $f(A \cup B) + f(A \cap B) \leq
f(A) + f(B)$. 

Submodular functions play a fundamental role in classical
combinatorial optimization where rank functions of matroids, edge
cuts, coverage, and others are instances of submodular functions (see
\cite{Schrijver_book,Fujishige_book}). More recently, there is a large
interest in constrained submodular function optimization driven both
by theoretical progress and a variety of applications in computer
science. The needs of the applications, and in particular the sheer
bulk of large data sets, have brought into focus the development of
fast algorithms for submodular optimization. Recent work on the
theoretical side include the development of faster worst-case
approximation algorithms in the traditional sequential model of
computation \cite{bv-famsf-14,ijbICML-13,ChekuriTV15}, algorithms in
the streaming model \cite{bmkk-sso-14,ck-smms-14} as well as in the
map-reduce model of computation \cite{kmvv-13}.

In this paper we consider constrained submodular function
\emph{maximization}. The goal is to find $\max_{S \in \independents}
f(S)$ where $\independents \subseteq \powerset{\groundset}$ is a
\newterm{downward-closed} family of sets; i.e., $A \in \independents$
and $B \subseteq A$ implies $B \in \independents$. $\independents$ is
also called an \newterm{independence family} and any set $A \in
\independents$ is called an \newterm{independent set}. Submodular
maximization under various independence constraints has been
extensively studied in the literature. The problem can be easily seen
to be NP-hard even for a simple cardinality constraint as it
encompasses standard NP-hard problems like the Max-$k$-cover
problem. Constrained submodular maximization has found several new
applications in recent years. Some of these include data summarization
\cite{hb-11,sssj-12,dkr-13}, influence maximization in social networks
\cite{kkt-03,cyy-09,ccy-10,gfl-11,ss-13}, generalized
assignment\cite{ccpv-07}, mechanism design \cite{bik-07}, and network
monitoring \cite{lkgfvg-07}.

In some of these applications, the amount of data involved is much
larger than the main memory capacity of individual computers. This
motivates the design of space-efficient algorithms which can process
the data in \textit{streaming} fashion, where only a small fraction of
the data is kept in memory at any point. There has been some recent
work on submodular function maximization in the streaming model,
focused on \emph{monotone} functions (i.e. $f(A) \leq f(B)$, whenever
$A \subseteq B$). This assumption is restrictive from both a
theoretical and practical point of view.

\begin{wrapfigure}{r}{.45\textwidth}
  \centering
  \includegraphics[width=.45\textwidth]{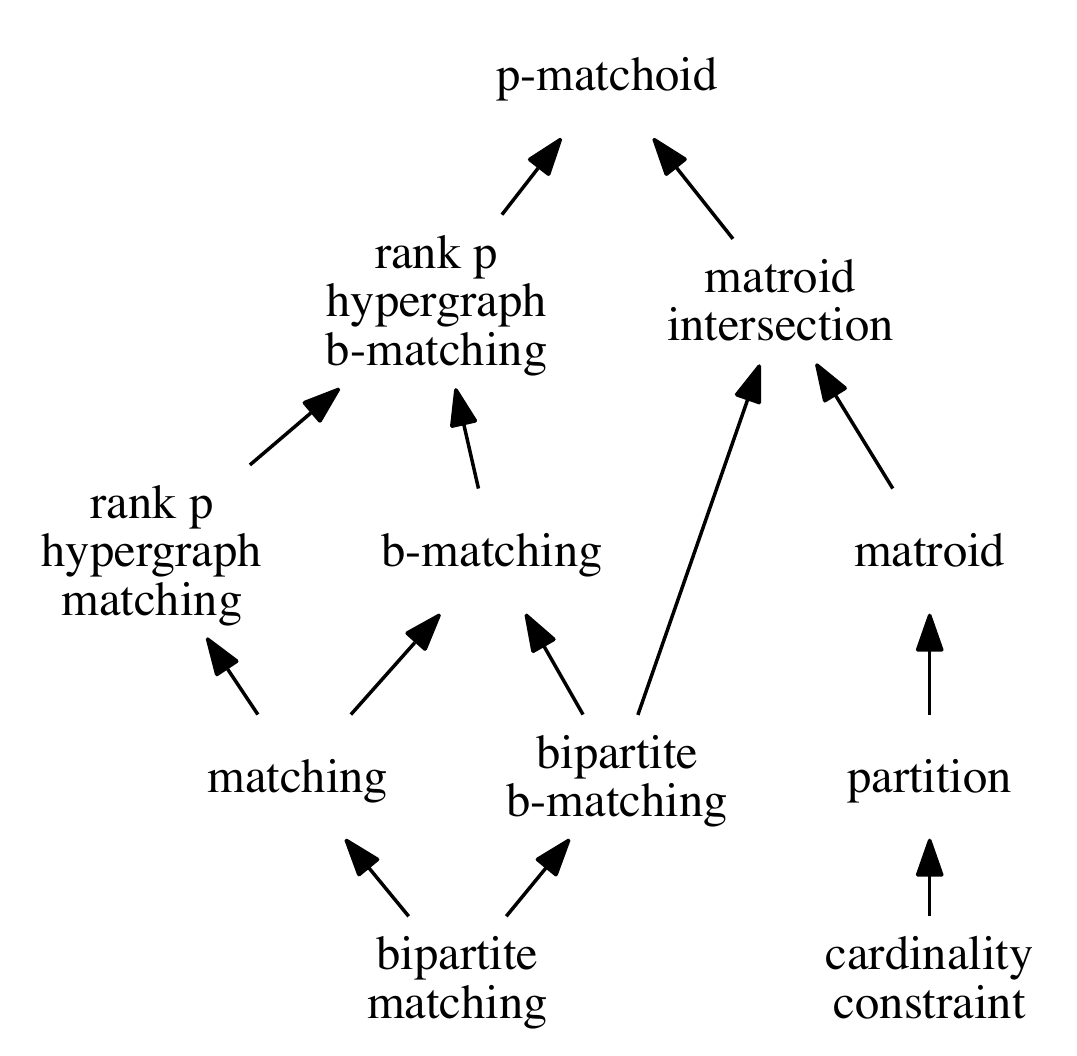}
  \caption{Hierarchy of set systems}
  \labelfigure{constraints}
\end{wrapfigure}
In this paper we present streaming algorithms for non-monotone
submodular function maximization subject to various combinatorial
constraints, the most general being a \refterm{$p$-matchoid}.
\refterm[$p$-matchoid]{$p$-matchoid's} generalize many basic
combinatorial constraints such as the cardinality constraint, the
intersection of $p$ matroids, and matchings in graphs and
hyper-graphs.  
A formal definition of a \refterm{$p$-matchoid} is given in
\refsection{preliminaries}. We consider the abstract
\refterm{$p$-matchoid} constraint for theoretical reasons, and most
constraints in practice should be simpler. We explicitly consider the
cardinality constraint and obtain an improved bound.

We now describe the problem formally. We are presented a groundset of
elements $\groundset = \{ e_1, e_2, \ldots e_n \}$, with no assumption
made on the order or the size of the datastream. The goal is to select
an independent set $S \subseteq \groundset$ (where independence is
defined by the \refterm[$p$-matchoid]{$p$-matchoid}), which maximizes
a nonnegative submodular function $f$ while using as little space as
possible. We make the following assumptions: (i) the function $f$ is
available via a value oracle, that takes as input a set $S \subseteq
\groundset$ and returns the value $f(S)$; (ii) the independence family
$\independents$ is available via a membership oracle with some
additional information needed in the \refterm{$p$-matchoid} setting;
and (iii) the constraints specify explicitly, and a priori, an upper
bound $k$ on the number of elements to be chosen. We discuss these in
turn.  The availability of a value oracle for $f$ is a reasonable and
standard assumption in the sequential model of computation, but needs
some justification in restrictive models of computation such as
streaming where the goal is to store at any point of time only a small
subset of the elements of $\groundset$. Can $f(S)$ be evaluated
without having access to all of $\groundset$? This of course depends
on $f$. \cite{bmkk-sso-14} gives several examples of interesting and
useful functions where this is indeed possible. The second assumption
is also reasonable if, as we remarked, the \refterm{$p$-matchoid}
constraint is in practice going to be a simple one that combines basic
matroids such as cardinality, partition and laminar matroid
constraints that can be specified compactly and implicitly. Finally,
the third assumption is guided by the fact that an abstract model of
constraints can in principle lead to every element being chosen. In
many applications the goal is to select a small and important subset
of elements from a much larger set; and it is therefore reasonable to
expect knowledge of an upper bound on how many can be chosen.
Submodular set functions are ubiquitous and arise explicitly and
implicitly in a variety of settings. The model we consider in this
paper may not be useful directly in some important scenarios of
interest. Nevertheless, the ideas underlying the analysis in the
streaming model that we consider here may still be useful in speeding
up existing algorithms and/or reduce their space usage.

As is typical for streaming algorithms, we measure performance in four
basic dimensions: (i) the approximation ratio $f(S)/\opt$, where $S$
is the output of the algorithm and $\opt$ is the value of an optimal
solution; (ii) the space usage of the algorithm; (iii) the update time
or the time required to process each stream element; and (iv) the
number of passes the algorithm makes over the data stream.

\begin{table}[t]
  \begingroup
  \footnotesize \def\arraystretch{1.6}
  \newcolumntype{C}{>{\centering\arraybackslash}X}%
  \newcolumntype{D}[1]{>{\hsize=#1\hsize\centering\arraybackslash}X}%
  \begin{tabularx}{\textwidth}{ | D{.9} | D{1} | D{1.2} | D{.8} |
      D{1.1} | }
    \cline{2-5} %
    \multicolumn{1}{ c| }{} & \multicolumn{2}{ |c| }{offline} %
    & \multicolumn{2}{ |c| }{streaming} \\
    \hline constraint %
    & monotone & nonnegative & monotone & nonnegative \\
    \hline %
    cardinality %
    & $1 - 1/e$ \cite{nwf-mssf1-78} %
    & $1/e + .004$ \cite{bfns-smcc-14} %
    & $\frac{1-\epsilon}{2} \cite{bmkk-sso-14}$ %
    & $\frac{1-\epsilon}{2 + e}$ (R,$\star$)\\
    \hline %
    matroid %
    & $1 - 1/e$ (R) \cite{ccpv-11} %
    & $\frac{1-\epsilon}{e}$ (R) \cite{fns-ucga-11} %
    & $1/4$ \cite{ck-smms-14} %
    & $\frac{1-\epsilon}{4+e}$ (R,$\star$) %
    \\ %
    \hline %
    matchings %
    & $\frac{1}{2+\epsilon}$ \cite{fnsw-iakes-11}%
    & $\frac{1}{4+\epsilon}$ \cite{fnsw-iakes-11}%
    & $4/31$ \cite{ck-smms-14}%
    & $\frac{1-\epsilon}{12 + \epsilon}$ (R,$\star$) %
    \\
    \hline %
    $b$-matchings %
    & $\frac{1}{2+\epsilon}$ \cite{fnsw-iakes-11} %
    & $\frac{1}{4 + \epsilon}$ \cite{fnsw-iakes-11} %
    & $1/8$ ($\star$) %
    & $\frac{1-\epsilon}{12 + \epsilon}$ (R,$\star$)%
    \\ %
    \hline %
    rank $p$ \mbox{hypergraph} $b$-matching %
    & $\frac{1}{p+\epsilon}$ (R) \cite{fnsw-iakes-11} %
    & $\frac{p-1}{p^2 + \epsilon}$ \cite{fnsw-iakes-11} %
    & $1/4p$ ($\star$) %
    & $\frac{(1-\epsilon)(p-1)}{5p^2 - 4p + \epsilon}$ (R,$\star$)
    \\
    \hline %
    intersection of \mbox{$p$ matroids} %
    & $\frac{1}{p+\epsilon}$ \cite{lsv-smmm-10} %
    & $\frac{p-1}{p^2 + (p-1)\epsilon}$ \cite{lsv-smmm-10}%
    & $1/4p$ \cite{ck-smms-14} %
    & $\frac{(1-\epsilon)(p-1)}{5p^2 - 4p}$ (R,$\star$)
    \\
    \hline %
    $p$-matchoids %
    & $\frac{1}{p+1}$ \cite{fnw-mssf2-78,ccpv-11} %
    & 
    $\frac{(1-\epsilon)(2-o(1))}{ep}$ (R)
    \cite{fns-ucga-11,cvz-mrcrs-11}%
    & $1/4p$ ($\star$) %
    & $\frac{(1-\epsilon)(2-o(1))}{(8+e)p}$ (R,$\star$)
    \\
    \hline %
  \end{tabularx}
  \endgroup
  \vspace{1ex}
  \caption{Best known approximation bounds for submodular
    maximization. Bounds for randomized algorithms that hold in
    expectation are marked (R). For hypergraph $b$-matchings
    and matroid intersection, $p$ is fixed. In the results for
    $p$-matchoids, $o(1)$ goes to zero as $p$ increases. New bounds
    attained in this paper are marked ($\star$). All new bounds except
    for the cardinality constraint are the first bounds for their
    class. The best previous bound for the cardinality constraint is
    about .0893, by \cite{bfs-osmp-15}.}
  \labeltable{results}
\end{table}

\paragraph{Our results.} We develop randomized and deterministic
algorithms that yield an $\Omega(1/p)$-approximation for maximizing a
non-negative submodular function under a \refterm{$p$-matchoid}
constraint in the one-pass streaming setting. The space usage is $O(k
\log k)$, essentially matching recent algorithms for the simpler
setting of maximizing a monotone submodular function subject to a
cardinality constraint \cite{bmkk-sso-14}. The randomized algorithm
achieves better constants than the deterministic algorithm. As far as
we are aware, we present the first streaming algorithms for
non-monotone submodular function maximization under constraints beyond
cardinality.  We give an improved bound of $\frac{1-\eps}{2+e}$ for
the cardinality constraint.  For the monotone case our bounds match
those of Chakrabarti and Kale \cite{ck-smms-14} for a single pass; we
give a self-contained algorithm and analysis. \reftable{results}
summarizes our results for a variety of constraints.

\paragraph{A brief overview of techniques.} Streaming 
algorithms for constrained modular and submodular function
optimization are usually clever variations of the greedy algorithm,
which picks elements in iterations to maximize the gain in each
iteration locally while maintaining feasibility.  For monotone
functions, in the offline setting, greedy gives a
$1/(p+1)$-approximation for the \refterm{$p$-matchoid} constraint and
a $(1-1/e)$-approximation for the cardinality constraint
\cite{fnw-mssf2-78}. The offline greedy algorithm cannot be directly
implemented in streams, but we outline two different strategies that
are still greedy in spirit. 
For the cardinality constraint, Badanidiyuru \etal \cite{bmkk-sso-14}
designed an algorithm that adds an element to its running solution $S$
only if the marginal gain is at least a threshold of about
$\opt/2k$. Although the quantity $\opt/2k$ is not known a priori, they
show that it lies in a small and identifiable range, and can be
approximated with $O(\log k)$ well-spaced guesses. The algorithm then
maintains $O(\log k)$ solutions in parallel, one for each
guess. Another strategy from Chakrabarti and Kale \cite{ck-smms-14},
based on previous work for matchings \cite{fkmsz-gpssm-05,m-fgmds-05}
and matroid constraints \cite{abv-11} with modular weights, will
consider deleting elements from $S$ when adding a new element to $S$
is infeasible. More specifically, when a new element $e$ is
encountered, the algorithm finds a subset $C \subseteq S$ such that
$(S \setminus C) + e$ is feasible, and compare the gain $f((S
\setminus C) + e) - f(S)$ to a quantity representing the value that
$C$ adds to $S$. In the modular case, this may be the sum of weights
of elements in $C$; for monotone submodular functions, Chakrabarti and
Kale used marginal values, fixed for each element when the element is
added to $S$, as proxy weights instead.


The non-monotone case is harder because marginal values can be
negative even when $f$ is non-negative. The natural greedy algorithm
fails for even the simple cardinality constraint, and the best offline
algorithms for nonnegative submodular maximization are uniformly
weaker (see \reftable{results}).  To this end, we adapt techniques
from the recent work of Buchbinder \etal \cite{bfns-smcc-14} in our
randomized algorithm, and techniques from Gupta \etal \cite{grst-10}
for the deterministic version. Buchbinder \etal randomized the
standard greedy algorithm (for cardinality) by repeatedly gathering
the top (say) $k$ remaining elements, and then randomly picking only
one of them. We adapt this to the greedy setting by adding the top
elements to a buffer $B$ as they appear in the stream, and randomly
adding an element from $B$ to $S$ only when $B$ fills up. What remains
of $B$ at the end of the stream is post-processed by an offline
algorithm. Gupta \etal gave a framework for adapting any monotone
submodular maximization algorithm to nonnegative submodular functions,
by first running the algorithm once to generate one independent set
$S_1$, then running the algorithm again on the complement of $S_1$ to
generate a second set $S_2$, and running an unconstrained maximization
algorithm on $S_1$ to produce a third set $S_3$, finally returning the
best of $S_1$, $S_2$, and $S_3$. Our deterministic streaming algorithm
is a natural adaptation, piping the rejected elements of one instance
of a streaming algorithm directly into a second instance of the same
algorithm, and post-processing all the elements taken by the first
streaming instance. Both of our algorithms require that we limit the
number of elements ever added to $S$, which then limits the size of
the input for the post-processor. This limit is enforced by the idea
of additive thresholds from \cite{bmkk-sso-14} and a simple but subtle
notion of value that ensures the properties we desire.

\paragraph{Related work.}
\labelsection{sec:reltd_result}

There is substantial literature on constrained submodular function
optimization, and we only give a quick overview.  Many of the basic
problems are NP-Hard, so we will mainly focus on the development of
approximation algorithms.  The (offline) problem $\max_{S \in
  \independents} f(S)$ for various constraints has been extensively
explored starting with the early work of Fisher, Nemhauser, Wolsey on
greedy and local search algorithms
\cite{nwf-mssf1-78,fnw-mssf2-78}. Recent work has obtained many new
and powerful results based on a variety of methods including variants
of greedy \cite{grst-10,bfns-smcc-14,bfjs-12}, local search
\cite{LeeMNS10,lsv-smmm-10,FilmusW14}, and the multilinear relaxation
\cite{ccpv-11,KulikST13,BansalKNS12,cvz-mrcrs-11}. Monotone submodular
functions admit better bounds than non-monotone functions (see
\reftable{results}).  For a \refterm{$p$-matchoid} constraint, which
is our primary consideration, an $\Omega(1/p)$-approximation can be
obtained for non-negative functions. Recent work has also obtained new
lower bounds on the approximation ratio achievable in the oracle model
via the so-called symmetery gap technique \cite{Vondrak13}; this also
yields lower bounds in the standard computational models \cite{dv-12}.

Streaming algorithms for submodular functions are a very recent
phenomenon with algorithms developed recently for monotone submodular
functions \cite{bmkk-sso-14,ck-smms-14}. \cite{bmkk-sso-14} gives a
$1/2-\eps$ approximation for monotone functions under cardinality
constraint using $O(k \log k/\eps)$ space. \cite{ck-smms-14} focuses
on more general constraints like interesctions of $p$-matroids and
rank $p$ hypergraphs, giving an approximation of $1/4p$ using a single
pass. Their algorithm extends to multiple passes, with an
approximation bound of $1/(p+1+\epsilon)$ with $O(\epsilon^{-3} \log
p)$ passes.  The main focus of \cite{kmvv-13} is on the map-reduce
model although they claim some streaming results as well.

Related to the streaming models are two {\em online} models where
elements arrive in an online fashion and the algorithm is required to
maintain a feasible solution $S$ at all times; each element on arrival
has to be processed and any element which is discarded from $S$ at any
time cannot be added back later. Strong lower bounds can be shown in
this model and two relaxations have been considered.  In the
\newterm{secretary model}, the elements arrive according to a random
permutation of the ground set and an element added to $S$ cannot be
discarded later. In the secretary model, constant factor algorithms
are known for the cardinality constraint and some special cases of a
single matroid constraint \cite{grst-10,bhz-13}.  These algorithms
assume the stream is randomly ordered and their performance degrades
badly against adversarial streams; the best competitive ratio for a
single general matroid is $O(\log k)$ (where $k$ is the rank of the
matroid).  Recently, Buchbinder \etal \cite{bfs-osmp-15} considered a
different relaxation of the online model where \newterm{preemptions}
are allowed: elements added to $S$ can be discarded later. Algorithms
in the preemptive model are usually streaming algorithms, but the
converse is not true (although the one-pass algorithms in
\cite{ck-smms-14} are preemptive). For instance, the algorithm in
\cite{bmkk-sso-14} maintains multiple feasible solutions and our
algorithms maintain a buffer of elements neither accepted nor
rejected. The space requirement of an algorithm in the online model is
not necessarily constrained since in principle an algorithm is allowed
to keep track of all the past elements seen so far. The main result in
\cite{bfs-osmp-15}, as it pertains to this work, is a randomized
$0.0893$-competitive algorithm for cardinality constraints using
$O(k)$-space. As \reftable{results} shows, we obtain a
$(1-\eps)/(2+e)$-competitive algorithm for this case using $O(k \log
k/\eps^2)$-space.

\paragraph{Paper organization.} \refsection{preliminaries} reviews
combinatorial definitions and introduces the notion of incremental
values. \refsection{streaming-greedy} analyzes an algorithm that works
for monotone submodular functions, and
\refsection{randomized-streaming-greedy} adapts this algorithm to the
non-monotone case. In \refsection{iterated-streaming-greedy}, we give
a deterministic streaming algorithm with slightly weaker guarantees.

\section{Preliminaries}
\labelsection{preliminaries}

\paragraph{Matroids.}
A \newterm{matroid} is a finite set system $\matroid =
(\groundset,\independents)$, where $\groundset$ is a set and
$\independents \subseteq \powerset{\groundset}$ is a family of subsets
such that:%
\begin{inline_properties}
\item $\emptyset \in \independents$,
\item If $A \subseteq B \subseteq \groundset$, and $B \in
  \independents$, then $A \in \independents$,
\item If $A, B \in \independents$ and $|A| < |B|$, then there is an
  element $b \in B \setminus A$ such that $A + b \in \independents$.
\end{inline_properties}
In a matroid $\matroid = (\groundset,\independents)$, $\groundset$ is
called the \newterm{ground set} and the members of $\independents$ are
called \newterm{independent sets} of the matroid. The bases of
$\matroid$ share a common cardinality, called the \newterm{rank} of
$\matroid$.




\paragraph{Matchoids.}
Let $\matroid_1 = (\groundset_1,\independents_1),\dots,\matroid_\nmat
= (\groundset_\nmat, \independents_\nmat)$ be $\nmat$ matroids over
overlapping groundsets. Let $\groundset = \groundset_1 \cup \cdots
\cup \groundset_\nmat$ and
\begin{math}
  \independents = \setof{S \subseteq \groundset \suchthat S \cap
    \groundset_{\mindex} \in \independents_{\mindex} \text{ for all
    }\ell}.
\end{math}
The finite set system $\defpmatchoid$ is a \newterm{$p$-matchoid} if
for every element $e \in \groundset$, $e$ is a member of
$\groundset_{\mindex}$ for at most $p$ indices $\mindex \in [\nmat]$.
$p$-matchoids generalizes matchings and intersections of matroids,
among others (see \reffigure{constraints}).


\paragraph{Maximizing submodular functions under a $p$-matchoid
  constraint.}

Let $\groundset$ be a set of elements, $f: \powerset{\groundset} \to
\nnreals$ a nonnegative submodular function on $\groundset$, and
$\defpmatchoid$ a $p$-matchoid for some integer $p$. 
We want to approximate $\OPT = \max_{S \in \independents} f(S)$.
There are several polynomial-time approximation algorithms that give
an $\Omega(1/p)$-approximation for this problem, with better bounds
for simpler constraints (see \reftable{results}).
These algorithms are used as a black box called \newalgo{Offline},
with approximation ratio denoted by $\gamma_p$: if \refalgo{Offline}
returns $S \in \independents$, then $\evof{f(S)} \geq \gamma_p \OPT$
(possibly without expectation, if \refalgo{Offline} is deterministic).



\paragraph{Incremental Value.}

\newcommand{\incvalue}{\nu}
\newcommand{\incrementalvalueof}[3]{\incvalue\parof{#1,#2,#3}}
\newcommand{\incvalueof}{\incrementalvalueof}

Let $\groundset$ be a ground set, and let $f: \subsetsof{\groundset}
\to \reals$ be a submodular function. For a set $S \subseteq
\groundset$ and an element $e \in S$, what is the value that $e$ adds
to $S$? One idea is to take the margin $f_{S - e}(e) = f(S) - f(S -
e)$ of adding $e$ to $S - e$. However, because $f$ is not necessarily
modular, we can only say that $\sum_{e \in S} f_{S - e}(e) \leq f(S)$
without equality. It is natural to ask for a different notion of value
where the values of the parts sum to the value of the whole.

Let $\groundset$ be an \emph{ordered} set and $f:
\subsetsof{\groundset} \to \reals$ be a set function.  For a set $S
\subseteq \groundset$ and element $e \in \groundset$, the
\newterm{incremental value} of $e$ in $S$, denoted
$\incvalueof{f}{S}{e}$, is defined as
\begin{align*}
  \incvalueof{f}{S}{e} &= f_{S'}(e) \text{, where } S' = \setof{s \in
    S \suchthat s < e}.
\end{align*}
The key point of incremental values is that they capture the entire
value of a set. The following holds for \emph{any} set function.
\begin{lemma}%
  \labellemma{sum-incremental-values}%
  Let $\groundset$ be an ordered set, $f: \subsetsof{\groundset} \to
  \reals$ a set function, and $S \subseteq \groundset$ a set. Then
  \begin{math}
    f(S) = \sum_{e \in S} \incvalueof{f}{S}{e}.
  \end{math}
  \begin{proof}
    Enumerate $S = \setof{e_1,\dots,e_\ell}$ in order, and let $S_i =
    \setof{e_1,\dots,e_i}$ denote the first $i$ elements in $S$. We
    have,
    \begin{align*}
      \sum_{e_i \in S} \incvalueof{f}{S}{e_i} %
      = %
      \sum_{e_i \in S} f_{S_{i-1}}(e_i) %
      = %
      f(S). %
    \end{align*}
  \end{proof}
\end{lemma}
When $f$ is submodular, we have decreasing incremental values
analogous (and closely related) to decreasing marginal returns of
submodular function.
\begin{lemma}%
  \labellemma{decreasing-incremental-values}
  Let $S \subseteq T \subseteq \groundset$ be two nested subsets of an
  ordered set $\groundset$, let $f: \subsetsof{\groundset} \to \reals$
  be submodular, and let $e \in \groundset$. Then
  \begin{math}
    \incvalueof{f}{T}{e} \leq \incvalueof{f}{S}{e}.
  \end{math}
  \begin{proof}
    Let $S' = \setof{s \in S \suchthat s < e}$ and $T' = \setof{t \in
      T \suchthat t < e}$. Since $S \subseteq T$, clearly $S'
    \subseteq T'$. We have,
    \begin{align*}
      \incvalueof{f}{T}{e} = f_{T'}(e) \leq f_{S'}(e) =
      \incvalueof{f}{S}{e}
    \end{align*}
    where the inequality follows by submodularity.
  \end{proof}
\end{lemma}
The following is also an easy consequence of submodularity.
\begin{lemma}
  \labellemma{marginal-incremental-values-ineq}
  Let $\groundset$ be an ordered set of elements, let $f:
  \subsetsof{\groundset} \to \reals$ be a submodular function, $S,Z
  \subseteq \groundset$ two sets, and $e \in S$. Then
  \begin{math}
    \incvalueof{f_Z}{S}{e} \leq \incvalueof{f}{Z \cup S}{e}.
  \end{math}
  \begin{proof}
    Let $Z' = \setof{z \in Z \suchthat z < e}$ and $S' = \setof{s \in
      S \suchthat s < e}$. By submodularity, we have,
    \begin{align*}
      \incvalueof{f_Z}{S}{e} = f_{Z \cup S'}(e) \leq f_{Z' \cup S'}(e)
      = \incvalueof{f}{Z \cup S}{e}.
    \end{align*}
  \end{proof}
\end{lemma}
\FloatBarrier
\section{Streaming Greedy}%
\labelsection{streaming-greedy}

\newcommand{\finalset}{\tilde{S}}%
\newcommand{\setbefore}[1]{S_{#1}^-} %
\newcommand{\setafter}[1]{S_{#1}^+} %
\newcommand{\takens}{U}%
\newcommand{\taken}{u} %
\newcommand{\candidates}{C}
\newcommand{\candidatesfor}[1]{\candidates_{#1}}
\newcommand{\gain}{\delta} %
\newcommand{\gainof}[1]{\delta_{#1}} %
\newcommand{\deleted}{d} %
\newcommand{\deleteds}{\takens \setminus \finalset} %


\begin{figure}[t]
  \centering
  \begin{minipage}{6.9cm}
    \begin{framed}
      \begin{pseudocode}
        \begin{routine}{Streaming-Greedy}{$\alpha$,$\beta$}
          $S \gets \emptyset$\\
          while (stream is not empty) \+\\
          $e \gets $ next element in the stream\\
          $C \gets \refalgo{Exchange-Candidates}{$S$,$e$}$ \\%
          \commentcode{$C$ satisfies $S - C + e \in \independents$}\\
          if %
          \begin{math}
            f_S(e) %
            \geq %
            \alpha + (1 + \beta) \sum_{c \in C} \incvalueof{f}{S}{c} %
          \end{math} %
          \\
          \> $S \gets S \setminus C + e$ \\
          \< end while \- \\
          return $S$
        \end{routine} %
      \end{pseudocode}
      \vspace{.2ex}
    \end{framed}
  \end{minipage}
  \qquad
  \begin{minipage}{6.4cm}
    \begin{framed}
      \begin{pseudocode}
        \begin{routine}{Exchange-Candidates}{$S$,$e$}
          $C \gets \emptyset$\\
          for $\ell = 1,\dots, \mcard$ \\
          \> if $e \in \groundset_\mindex$ and $(S + e) \cap
          \groundset_\mindex
          \notin \independents_\mindex$\\
          \> \> $S_\ell = S \cap \groundset_\ell$\\
          \> \> $X \gets \setof{s \in S_\ell: (S_\ell - s + e) \in
            \independents_\ell}$\\
          \> \> \commentcode{$X + e$ is a circuit}\\
          \> \> $c_{\mindex} \gets \argmin_{x \in X} \incvalueof{f}{S}{x}$\\
          \> \> $C \gets C + c_{\mindex}$ \\
          \> end if \\
          end for\\
          return $C$
        \end{routine}
      \end{pseudocode}
    \end{framed}
  \end{minipage}
\end{figure}


Let $\defpmatchoid$ be a $p$-matchoid and $f$ a submodular
function. The elements of $\groundset$ are presented in a stream, and
we order $\groundset$ by order of appearance. We assume value oracle
access to $f$, that given $S \subseteq \groundset$, returns the value
$f(S)$. We also assume membership oracles for each of the $\mcard$
matroids defining $\pmatchoid$: given $S \subseteq \groundsetI$, there
is an oracle for $\matroidI$ that returns whether or not $S \in
\independentsI$.


We first present a deterministic streaming algorithm
\refalgo{Streaming-Greedy} that yields an
$\Omega(1/p)$-{\allowbreak}approximation for monotone submodular
functions, but performs poorly for non-monotone functions. The primary
motivation in presenting \refalgo{Streaming-Greedy} is as a building
block for a randomized algorithm \refalgo{Randomized-Streaming-Greedy}
presented in \refsection{randomized-streaming-greedy}, and a
deterministic algorithm \refalgo{Iterated-Streaming-Greedy} presented
in \refsection{iterated-streaming-greedy}. The analysis for these
algorithms relies crucially on properties of
\refalgo{Streaming-Greedy}.

\refalgo{Streaming-Greedy} maintains an independent set $S \in
\independents$; as an element arrives in the stream, it is either
discarded or added to $S$ in exchange for a well-chosen subset of $S$.
The threshold for exchanging is tuned by two nonnegative parameters
$\alpha$ and $\beta$.  At the end of the stream,
\refalgo{Streaming-Greedy} outputs $S$.

The overall strategy is similar to previous algorithms developed for
matchings \cite{fkmsz-gpssm-05,m-fgmds-05} and intersections of
matroids \cite{abv-11} when $f$ is modular, and generalized by
\cite{ck-smms-14} to monotone submodular functions.  There are two
main differences. One is the use of the additive threshold
$\alpha$. The second is the use of the incremental value $\nu$. By
using incremental value, the value of an element $e \in S$ is not
fixed statically when $e$ is first added to $S$, and increases over
time as other elements are dropped from $S$. These two seemingly minor
modifications are crucial to the eventual algorithms for non-monotone
functions.


We remark that \refalgo{Streaming-Greedy} also fits the online
preemptive model.

\newcommand{\replacement}{e(\deleted)} %

\paragraph{Outline of the analysis:} Let $T \in \independents$ be some
fixed feasible set (we can think of $T$ as an optimum set).  In the
offline analysis of the standard greedy algorithm one can show that
$f(S \cup T) \le (p+1) f(S)$, where $S$ is the output of greedy; for
the monotone case this implies that $f(S) \ge f(T)/(p+1)$.  The
analysis here hinges on the fact that each element of $T \setminus S$
is available to greedy when it chooses each element.  In the streaming
setting, this is no longer feasible and hence the need to remove
elements in favor of new high-value elements.  To relate $\tilde{S}$,
the final output, to $T$, we consider $U$, the set of all elements
ever added to $S$. The analysis proceeds in two steps.

First, we upper bound $f(U)$ by $f(\tilde{S})$ as 
\begin{math}
  f(U) %
  \leq %
  (1+ \frac{1}{\beta}) \cdot f(\finalset) - \frac{\alpha}{\beta}
  \sizeof{U}.
\end{math}
Second, we upper bound $f(T \cup \takens)$ as 
\begin{math}
  f(T \cup \takens) \leq k \alpha %
  + %
  \frac{(1 + \beta)^2}{\beta} \cdot p \cdot f(\finalset).
\end{math}
For $\alpha = 0$, we obtain $f(T \cup \takens) \le \frac{(1 +
  \beta)^2}{\beta} \cdot p \cdot f(\finalset)$, which yields $f(T) \le
4p f(\finalset)$ when $f$ is monotone (for $\beta = 1$); this gives
the same bound as \cite{ck-smms-14}.  The crucial difference is that
we are able to prove an upper bound on the size of $U$, namely,
$|\takens| \le \opt/\alpha$; hence, if we choose the threshold
$\alpha$ to be $c \opt/k$ for some parameter $c$ we have $|U| \le
k/c$. This will play a critical role in analyzing the non-monotone
case in the subsequent sections that use $\refalgo{Streaming-Greedy}$
as a black box. The upper bound on $|U|$ is achieved by the definition
of $\incvalue$ and the threshold $\alpha$; we stress that this is not
as obvious as it may seem because the function $f$ can be non-monotone
and the marginal values can be negative.

\paragraph{Some notation for the analysis:}
\begin{itemize}
\item $\finalset$ denotes the final set returned by
  \refalgo{Streaming-Greedy}.
\item For each element $e \in \groundset$, $\setbefore{e}$ denotes the
  set held by $S$ just before $e$ is processed, and $\setafter{e}$ the
  set held by $S$ just after $e$ is processed. Note that if $e$ is
  rejected, then $\setbefore{e} = \setafter{e}$.
\item $\takens$ denotes the set of all elements added to $S$ at any
  point in the stream. Note that $\takens = \bigcup_{e \in \groundset}
  \setafter{e}$.
\item For $e \in \groundset$, $\candidatesfor{e} =
  \refalgo{Exchange-Candidates}{$\setbefore{e}$,$e$} \subseteq
  \setbefore{e}$ denotes the set of elements that
  $\refalgo{Streaming-Greedy}$ considers exchanging for $e$. Observe
  that $\setof{\candidatesfor{\taken},\taken \in \takens}$ forms a
  partition of $\takens \setminus \finalset$.
\item For $e \in \groundset$, $\gainof{e} = f(\setafter{e}) -
  f(\setbefore{e})$ denotes the {\em gain} from processing $e$. Note
  that $\gainof{e} = 0$ for all $e \in \groundset \setminus \takens$,
  and $\sum_{e \in \groundset} \gainof{e} = f(\finalset)$.
\end{itemize}

\subsection{Relating $f(\takens)$ to $f(\finalset)$}
\labelsection{takens-vs-final-set}
\newcommand{\exitvalueof}{\chi\parof}

When \refalgo{Streaming-Greedy} adds an element $e$ to $S$, it only
compares the marginal $f_S(e)$ to the incremental values in its
exchange candidates $C$, and does not directly evaluate the gain $f(S
\setminus C + e) - f(S)$ realized by the exchange. The first lemma
derives a lower bound for this gain.
\begin{lemma}%
  \labellemma{gain-over-incremental-value-lost}
  Let $e \in \takens$ be added to $S$ when processed by
  \refalgo{Streaming-Greedy}. Then
  \begin{align*}
    \gainof{e} %
    \geq %
    \alpha + \beta \sum_{c \in \candidatesfor{e}}
    \incvalueof{f}{\setbefore{e}}{c}
  \end{align*}
\end{lemma}
\begin{proof}
  Since $e$ replaced $\candidatesfor{e}$, by design of
  \refalgo{Streaming-Greedy}, we have,
  \begin{align*}
    f_{\setbefore{e}}(e) %
    \geq %
    \alpha + (1+\beta) \sum_{c \in \candidatesfor{e}}
    \incvalueof{f}{\setbefore{e}}{c},
  \end{align*}
  which, after rearranging, gives
  \begin{align*}
    f_{\setbefore{e}}(e) - \sum_{c \in \candidatesfor{e}}
    \incvalueof{f}{\setbefore{e}}{c}%
    \geq %
    \alpha + \beta \sum_{c \in \candidatesfor{e}}
    \incvalueof{f}{\setbefore{e}}{c}.
  \end{align*}
  To prove the lemma, it suffices to show that
  \begin{align*}
    \gainof{e} %
    \geq %
    f_{\setbefore{e}}(e) - \sum_{c \in \candidatesfor{e}}
    \incvalueof{f}{\setbefore{e}}{c}.
  \end{align*}
  Let $Z = \setbefore{e} \setminus \candidatesfor{e} = \setafter{e} -
  e$. Note that $\setafter{e} = Z + e$ and $\setbefore{e} = Z \cup
  \candidatesfor{e}$. We have,
  \begin{align*}
    \gainof{e} %
    &= %
    f(Z + e) - f(Z + \candidatesfor{e}) \\
    &= %
    f_Z(e) - f_Z(\candidatesfor{e}) %
    &\text{by adding and
      subtracting } f(Z), \\
    &\geq %
    f_{\setbefore{e}}(e) - f_Z(\candidatesfor{e}) &\text{by
      submodularity of $f$,} \\
    &= %
    f_{\setbefore{e}}(e) - \sum_{c \in \candidatesfor{e}}
    \incvalueof{f_Z}{\candidatesfor{e}}{c}, %
    &\text{by definition of } \incvalue.
    \\
    &\geq %
    f_{\setbefore{e}}(e) - \sum_{c \in
      \candidatesfor{e}}\incvalueof{f}{\setbefore{e}}{c} %
    &\text{by \reflemma{marginal-incremental-values-ineq} and
      $\setbefore{e} = Z\cup C_e$,}
  \end{align*}
  as desired.
\end{proof}
One basic consequence of \reflemma{gain-over-incremental-value-lost}
is that every element in $\takens$ adds a positive and significant
amount $\alpha$ to the value to $S$. This will be crucial later, when
taking $\alpha$ proportional to $\OPT$ limits the size of $\takens$.
\begin{lemma}%
  \labellemma{size-of-tokens} %
  For all $e \in \takens$, $\gainof{e} \geq \alpha$ and hence
  $\sizeof{\takens} \leq \OPT / \alpha$.
\end{lemma}
\begin{proof}
  We claim that at any point in the algorithm, $\incvalueof{f}{S}{e}
  \geq 0$ for all $e \in S$, from which the lemma follows
  \reflemma{gain-over-incremental-value-lost} immediately.

  When an element $e \in \takens$ is added to $S$, it has incremental
  value
  \begin{align*}
    \incvalueof{f}{\setafter{e}}{e} %
    = %
    f_{\setafter{e} - e}(e) %
    \geq %
    f_{\setbefore{e}}(e) %
    \geq %
    \alpha + (1 + \beta) \sum_{c \in \candidatesfor{e}}
    \incvalueof{f}{S}{c}.
  \end{align*}
  As the algorithm continues, elements preceding $e$ in $S$ may be
  deleted while elements after $e$ are added, so
  $\incvalueof{f}{S}{e}$ can only increase with time.
\end{proof}
Returning to the original task of bounding $f(U)$, the difference
$\takens \setminus \finalset$ is the set of deleted elements, and the
only handle on these elements is their incremental value at the point
of deletion. For $\deleted \in \deleteds$, let $\replacement$ be the
element that $\deleted$ was exchanged for; that is, $\replacement > d$
and $d \in \setbefore{\replacement} \setminus \setafter{\replacement}
= \candidatesfor{\replacement}$. For deleted elements $\deleted \in
\deleteds$, the \newterm{exit value} $\exitvalueof{\deleted}$ of
$\deleted$ is the incremental value of $\deleted$ evaluated when
$\deleted$ is removed from $S$, defined formally as
\begin{align*}
  \exitvalueof{\deleted} %
  &= %
  \incvalueof{f}{\setbefore{\replacement}}{\deleted}.
\end{align*}
Here we bound the sum of exit values of $\deleteds$.
\begin{lemma}
  \labellemma{sum-exit-values}
  \begin{align*}
    \sum_{\deleted \in \deleteds} \exitvalueof{\deleted} %
    \leq %
    \frac{1}{\beta} \cdot \parof{f(\finalset) - \alpha \sizeof{U}}.
  \end{align*}
\end{lemma}
\begin{proof}
  Indeed,
  \begin{align*}
    \sum_{\deleted \in \deleteds} \exitvalueof{\deleted} %
    &= %
    \sum_{\taken \in \takens} \sum_{\deleted \in
      \candidatesfor{\taken}} \exitvalueof{d} %
    &\text{since } \setof{\candidatesfor{\taken} \where \taken \in
      \takens} \text{ partitions } \deleteds,\\
    &\leq %
    \sum_{\taken \in \takens} \frac{1}{\beta}
    \cdot \parof{\gainof{\taken} -\alpha}%
    &\text{by \reflemma{gain-over-incremental-value-lost},}\\
    &= %
    \frac{1}{\beta} \cdot \parof{f(\finalset) - \alpha \sizeof{U}}
  \end{align*}
\end{proof}
Now we bound $f(U)$.
\begin{lemma}
  \labellemma{takens-bound}
  \begin{align*}
    f(U) %
    \leq %
    \parof{1 + \frac{1}{\beta}} \cdot f(\finalset) %
    - \frac{\alpha}{\beta} \sizeof{U}.
  \end{align*}
\end{lemma}
\begin{proof}
  Recall, for each element $\deleted \in U \setminus \finalset$ ,
  $\replacement$ denotes the element added in exchange of
  $\deleted$. We have,
  \begin{align*}
    f(U) - f(\finalset)%
    &= %
    f_{\finalset}(U) = \sum_{\deleted \in U \setminus \finalset} %
    \incvalueof{f_{\finalset}}{U}{\deleted}
    &\text{by \reflemma{sum-incremental-values},} \\
    &\leq %
    \sum_{\deleted \in U \setminus \finalset} \exitvalueof{\deleted} %
    &\text{by \reflemma{decreasing-incremental-values},}\\
    &\leq %
    \frac{1}{\beta} \cdot f(\finalset) -
    \frac{\alpha}{\beta}\sizeof{U}%
    &\text{by \reflemma{sum-exit-values}.}
  \end{align*}
\end{proof}
\begin{remark}
  The preceding lemmas relating $f(\takens)$ and $f(\finalset)$ do not
  rely on the structure of $\defpmatchoid$.
\end{remark}
\subsection{Upper bounding $f(U \cup T)$}
\labelsection{takens-and-competition-vs-takens} %
Let $T \in \independents$ be any feasible solution. The goal is to
upper bound $f(U \cup T)$. Here we use the fact that $\independents$
is a $p$-matchoid to frame an exchange argument between $T$ and $U$.

\begin{lemma}%
  \labellemma{nn-submodular-streaming-exchange-lemma}
  Let $T \in \independents$ be a feasible solution \emph{disjoint}
  from $\takens$. There exists a mapping $\varphi: T \to
  \subsetsof{\takens}$ such that
  \begin{results}
  \item Every $s \in \finalset$ appears in the set $\varphi(t)$ for at
    most $p$ choices of $t \in T$.
  \item Every $\deleted \in \deleteds$ appears in the set $\varphi(t)$
    for at most $(p-1)$ choices of $t \in T$.
  \item For each $t \in T$,
    \begin{align}%
      \labelequation{matchoid-exchange-inequality}%
      \sum_{c \in \candidatesfor{t}}
      \incvalueof{f}{\setbefore{t}}{c} %
      &\leq %
      \sum_{\deleted \in \varphi(t) \setminus \finalset}
      \exitvalueof{\deleted}%
      +%
      \sum_{s \in \varphi(t) \cap
        \finalset}\incvalueof{f}{\finalset}{s}.
    \end{align}
  \end{results}
\end{lemma}

\begin{proof}
  The high level strategy is as follows. For each matroid
  $\defmatroidI$ in the $p$-matchoid $\pmatchoid$, we construct a
  directed acyclic graph $\graph_{\mindex}$ on $\groundsetI$, where a
  subset of $T$ forms the source vertices and arrows preserve
  \refequation[inequality]{matchoid-exchange-inequality}. Applying
  \reflemma{graphical-span-matching} we get an injection from a subset
  of $T$ into $U \cap \groundsetI$. With care, the union of these
  injections will produce the mapping we seek.

  Let us review and annotate the subroutine
  $\refalgo{Exchange-Candidates}{$S$,$e$}$. For each matroid
  $\defmatroidI$ in which $S$ spans $e$ (i.e., $(S + e) \cap
  \groundsetI \notin \independentsI$), we assemble a subset
  $X_{e,\mindex} \subseteq S \cap \groundsetI$ that spans $e$ in
  $\matroidI$. Of these, we choose the element $c_{e,\mindex} \in
  X_{e,\mindex}$ with the smallest incremental value with respect to
  $S$.

  Fix a matroid $\defmatroidI$. Let
  \begin{align*}
    T_{\mindex} %
    = %
    \setof{%
      t \in T \cap \groundsetI %
      \suchthat %
      (\setbefore{t} + t) \cap \groundsetI \notin \independentsI %
    }
  \end{align*}
  be the set of elements in $T$ obstructed by $\matroidI$, so to
  speak. For each $x \in X_{t, \mindex}$, add a directed edge
  $\arc{t}{x}$ from $t$ to $x$. Observe that for all $t \in
  T_{\mindex}$, $\outneighborhoodof[\graph_{\mindex}]{t} = X_{t,
    \mindex}$ spans $t$.

  Let
  \begin{align*}
    D_{\mindex} %
    = %
    \setof{\deleted \suchthat d = c_{e,\mindex} \text{ for some } e
      \in U}
  \end{align*}
  be the elements deleted specifically for $\matroidI$.  Observe that
  $D_{\mindex} \subseteq (\deleteds) \cap \groundsetI$. For $\deleted
  \in D_{\mindex} $, the set $Y_{\deleted,\mindex} =
  X_{\replacement,\mindex} - \deleted + \replacement$ spans
  $\deleted$, and for all $y \in Y_{\deleted,\mindex}$,
  \begin{align*}
    \exitvalueof{\deleted} %
    = %
    \incvalueof{f}{\setbefore{\replacement}}{\deleted} %
    \leq %
    \incvalueof{f}{\setbefore{\replacement}}{y} %
    \leq %
    \begin{cases}
      \exitvalueof{y} & \text{if } y \in \deleteds, \\
      \incvalueof{f}{\finalset}{y} & \text{if }y \in \finalset.
    \end{cases}
  \end{align*}
  For each $y \in Y_{\deleted,\mindex}$, add the directed edge
  $\arc{\deleted}{y}$ from $\deleted$ to $y$. Observe that for all
  $\deleted \in D_{\mindex}$,
  $\outneighborhoodof[\graph_{\mindex}]{\deleted} =
  Y_{\deleted,\mindex}$ spans $\deleted$.

  Clearly, $\graph_{\ell}$ is a directed acyclic graph. The elements
  of $T_{\mindex}$ are sources in $\graph_\ell$, and the elements of
  $D_{\mindex}$ are never sinks. By
  \reflemma{graphical-span-matching}, there exists an injection
  $\varphi_{\mindex}$ from $T_{\mindex}$ to $(\takens \cap
  \groundsetI) \setminus D_{\mindex}$ such that for each $t \in
  T_{\mindex}$, there is a path in $\graph_{\mindex}$ from $t$ to
  $\varphi_{\mindex}(t)$. If we write out the path $t \to x_1 \to
  \cdots \to x_{z} \to \varphi_{\mindex}(t)$, we have $x_1,\dots,x_{z}
  \in \deleteds$ and
  \begin{align*}
    \incvalueof{f}{\setbefore{t}}{c_{t,\mindex}}%
    \leq \incvalueof{f}{\setbefore{t}}{x_1} %
    \leq \exitvalueof{x_1} %
    \leq \cdots \leq \exitvalueof{x_z} %
    \leq %
    \begin{cases}
      \exitvalueof{\varphi_{\ell}(t)} %
      & \text{if } \varphi_{\ell}(t) \in \deleteds, \\
      \incvalueof{f}{\finalset}{\varphi_{\ell}(t)} %
      & \text{if }\varphi_{\ell}(t) \in \finalset.
    \end{cases}
  \end{align*}

  After constructing $\varphi_{\mindex}$ for each matroid $\matroidI$,
  define $\varphi : T \to \subsetsof{\takens}$ by
  \begin{align*}
    \varphi(t) = \bigcup_{\ell: t \in T_{\mindex}}
    \varphi_{\mindex}(t).
  \end{align*}
  For each $t \in T$, we have
  \begin{align*}
    \sum_{c \in \candidatesfor{t}} \incvalueof{f}{\setbefore{t}}{c}%
    \leq %
    \sum_{\deleted \in \varphi(t) \setminus \finalset}
    \exitvalueof{\deleted} %
    + %
    \sum_{s \in \finalset \cap \varphi(t)}
    \incvalueof{f}{\finalset}{s}.
  \end{align*}
  Each $\taken \in \takens$ belongs to at most $p$ matroids, so each
  $\taken \in \takens$ appears in $\varphi(t)$ for at most $p$ values
  of $t \in T$.  Since $\setof{D_{\mindex}}$ covers $\deleteds$, and
  $\varphi_{\mindex}$ avoids $D_{\mindex}$, each $\deleted \in
  \deleteds$ appears in $\varphi(t)$ at most $p - 1$ times.
\end{proof}
\begin{remark}
  A similar exchange lemma is given by Badanidiyuru for the
  intersection of $p$ matroids with modular weights \cite{abv-11}, and
  used implicitly by Chakrabarti and Kale in their extension to
  submodular weights. Here we extend the argument to $p$-matchoids and
  frame it in terms of incremental values.
\end{remark}

Now we bound $f(T \cup \takens)$.
\begin{lemma}
  Let $T \in \independents$ be an independent set. Then
  \begin{align*}
    f(T \cup \takens) \leq k \alpha %
    + %
    \frac{(1 + \beta)^2}{\beta} \cdot p \cdot f(\finalset).
  \end{align*}
\end{lemma}
\begin{proof}
  Let $T' = T \setminus \takens$. By submodularity, we have,
  \begin{align*}
    f_U(T) &\leq %
    \sum_{t \in T'} f_U(t) %
    \leq %
    \sum_{t \in T'} f_{\setbefore{t}}(t).
  \end{align*}
  Since each $t \in T'$ is rejected, and $\sizeof{T'} \leq k$, we have
  \begin{align*}
    \sum_{t \in T'} f_{\setbefore{t}}(t) %
    \leq %
    (1+\beta) \sum_{t \in T'} \sum_{c \in \candidatesfor{t}}
    \incvalueof{f}{\setbefore{t}}{c} %
    + \alpha k.  %
  \end{align*}
  Apply \reflemma{nn-submodular-streaming-exchange-lemma} to generate
  a mapping $\varphi: T' \to \subsetsof{\takens}$. We have,
  \begin{align*}
    \lhs %
    (1+\beta) \sum_{t \in T'} \sum_{c \in \candidatesfor{t}}
    \incvalueof{f}{\setbefore{t}}{c} \\%
    &\leq %
    (1+\beta)\sum_{t \in T'} \parof{%
      \sum_{\deleted \in \varphi(t) \setminus \finalset}
      \exitvalueof{\deleted} %
      + %
      \sum_{s \in \finalset \cap \varphi(t)}
      \incvalueof{f}{\finalset}{s}%
    } %
    &\text{by construction of $\varphi$,}\\
    &\leq %
    (1+\beta) \cdot (p-1) \sum_{d \in \deleteds} \exitvalueof{d} +
    (1+\beta) \cdot p \sum_{s \in \finalset}
    \incvalueof{f}{\finalset}{s} %
    &\text{by construction of $\varphi$,}\\
    &\leq %
    \frac{1+\beta}{\beta} \cdot (p-1) \cdot f(\finalset) + p \cdot (1
    + \beta) \cdot f(\finalset) %
    &\text{by \reflemma{sum-exit-values},}\\ %
    &= %
    \parof{\frac{(1+\beta)^2}{\beta} \cdot p -
      \frac{1+\beta}{\beta}}f\parof{\finalset}.
  \end{align*}
  To bound $f(\takens \cup T)$, we have,
  \begin{align*}
    f(\takens \cup T) %
    &= f_{\takens}(T) + f(\takens) \\
    &\leq %
    k \alpha + \parof{\frac{(1 + \beta)^2}{\beta} \cdot p -
      \frac{1+\beta}{\beta}} f (\finalset) + f(\takens) %
    &\text{by the above,} \\
    &\leq k \alpha + \frac{(1 + \beta)^2}{\beta} \cdot p \cdot f
    (\finalset) &\text{by \reflemma{takens-bound}},
  \end{align*}
  as desired.
\end{proof}

\subsection{A bound for the monotone case}
\labelsection{monotone-competition-vs-final-set} %
If $f$ is monotone, then $f(T) \leq f(\takens \cup T)$ for any set
$T$. If we take $T$ to be the set $\optset$ achieving $\OPT$, $\alpha
= 0$, and $\beta = 1$, we obtain the followings.
\begin{corollary}%
  \labelcorollary{monotone-bounds}
  Let $\defpmatchoid$ be a $p$-matchoid of rank $k$, and let $f:
  \subsetsof{\groundset} \to \nnreals$ be a nonnegative monotone
  submodular function. Given a stream over $\groundset$,
  \refalgo{Streaming-Greedy}{$0$,$1$} is an online algorithm that
  returns a set $\finalset \in \independents$ such that
  \begin{align*}
    f(\finalset) \geq \frac{1}{4p} \cdot \OPT = \frac{1}{4p} \cdot
    \maxof{f(T) \where T \in \independents}.
  \end{align*}
\end{corollary}
\begin{remark}
  Although these bounds match those of Chakrabarti and Kale for the
  intersection of $p$ matroids \cite{ck-smms-14},
  \refalgo{Randomized-Greedy} requires more calls to the submodular
  value oracle as the incremental value of an element in $S$ updates
  over time. That said, the number of times a taken element $e \in U$
  reevaluates its incremental value is proportional to the number of
  times an element in $\setbefore{e}$ is deleted, which is at most the
  rank of $\pmatchoid$ and generally considered small compared to the
  size of the stream.  Furthermore, by taking $\alpha$ proportional to
  $\OPT$ (a procedure for which is discussed in
  \refsection{estimating-threshold}), we can limit the size of $U$ and
  thereby the number of additional oracle calls generated by shifting
  incremental values.
\end{remark}


\section{Randomized Streaming Greedy}
\labelsection{randomized-streaming-greedy}

\newcommand{\buffer}{B}%
\newcommand{\bufferlimit}{K}%
\newcommand{\instance}{G}%
\newcommand{\finalbuffer}{\tilde{B}}%
\newcommand{\bestset}{\hat{S}}%

\refalgo{Randomized-Streaming-Greedy} adapts
\refalgo{Streaming-Greedy} to nonnegative submodular functions by
employing a randomized buffer $\buffer$ to limit the probability that
any element is added to the running solution $S$.  Like
\refalgo{Streaming-Greedy}, \refalgo{Randomized-Streaming-Greedy}
maintains the invariant $S \in \independents$. However, when a
``good'' element would have been added to $S$ by
\refalgo{Streaming-Greedy}, it is instead placed in $\buffer$. Once
the number of elements in $\buffer$ hits a limit $K$, we pick one
element in $\buffer$ uniformly at random and add it to $S$ just as
\refalgo{Streaming-Greedy} would.

\begin{figure}
  \centering
  \begin{minipage}{6.9cm}
    \begin{framed}
      \begin{pseudocode}
        \begin{routine}{Randomized-Streaming-Greedy$(\alpha,\beta)$}
          $S \gets \emptyset$, $B \gets \emptyset$ \\ 
          while (stream is not empty) \\
          \> $e \gets$ next element in the stream\\
          \> if $\refalgo{Is-Good}{$S$,$e$}$ then $\buffer \gets \buffer + e$ \\
          \> if $\sizeof{\buffer} = \bufferlimit$ then \+ \\
          \> $e \gets$ uniformly random from $\buffer$ \\
          \> $C \gets \refalgo{Exchange-Candidates}{$S$,$e$}$ \\%
          \> $\buffer \gets \buffer - e$, $S \gets (S \setminus C) + e$ \\
          \> for all $e' \in \buffer$ \\
          \> \> unless \refalgo{Is-Good}{$S$,$e'$}\\
          \> \> \> $\buffer \gets \buffer - e'$\\
          end if \\
          \< end while \- \\
          $S' \gets \refalgo{Offline}{$\buffer$}$ \\
          return $\argmax_{Z \in \setof{S,S'}} f(Z)$
        \end{routine}
      \end{pseudocode}
    \end{framed}
  \end{minipage}
  \qquad
  \begin{minipage}{6.8cm}
    \begin{framed}
      \begin{pseudocode}
        \begin{routine}{Is-Good}{$S$,$e$}
          $C \gets \refalgo{Exchange-Candidates}{$S$,$e$}$\\
          if $f_S(e) \geq \alpha + (1 + \beta) \sum_{e' \in C}
          \incvalueof{f}{S}{e'}$ \\
          \> return TRUE\\
          else return FALSE
        \end{routine}
      \end{pseudocode}
    \end{framed}
  \end{minipage}
\end{figure}


Modifying $S$ may break the invariant that the buffer only contains
good elements. Since $f$ is submodular, the incremental value
$\incvalueof{f}{S}{e}$ of each $e \in S$ may increase if a preceding
element is deleted. Furthermore, the marginal value $f_S(b)$ of each
buffered element $b \in \buffer$ may decrease as elements are added to
$S$. Thus, after modifying $S$, we reevaluate each $b \in \buffer$ and
discard elements that are no longer good.

Let $\finalbuffer$ be the set of elements remaining in the buffer
$\buffer$ when the stream ends.  We process $\finalbuffer$ with an
offline algorithm to produce a second solution $S'$, and finally
return the set $\bestset$ which is the better of $S$ and $S'$.

\paragraph{Outline of the analysis.}

Let $T \in \independents$ be an arbitrary independent set. Let $T' = T
\setminus \finalbuffer$ be the portion fully processed by the online
portion and $T'' = T \cap \finalbuffer$ the remainder left over in the
buffer and processed offline.

In \refsection{reduction-to-streaming-greedy}, we first show that the
analysis for $T'$ largely reduces to that of
\refsection{streaming-greedy}. In particular, this gives us a bound on
$f(U \cup T')$. In \refsection{online-plus-offline}, we combine this
with a bound on $f(T'')$, guaranteed by the offline algorithm, to
obtain an overall bound on $f(U \cup T)$ by $f(\bestset)$. In
\refsection{competition-vs-competition-and-takens}, we finally bound
$f(T)$ with respect to $f(U)$, leveraging the fact that the buffer
limits the probability of elements being added to $S$. In
\refsection{competition-vs-final-set}, we tie together the analysis to
bound $f(T)$ by $f(\bestset)$ for fixed $\alpha$ and $\beta$.

The analysis reveals that the optimal choice for $\beta$ is 1, and
that $\alpha$ should be chosen in proportion to $\OPT/k$, where $k$ is
the rank of the $\pmatchoid$. Since $\OPT$ is not known a priori, in
\refsection{estimating-threshold}, we leverage a technique by
Badanidiyuru \etal \cite{bmkk-sso-14} that efficiently guesses the
$\alpha$ to within a constant factor of the target value. The final
algorithm is then $\log k$ copies of
\refalgo{Randomized-Streaming-Greedy} run in parallel, each instance
corresponding to a ``guess'' for $\alpha$. One of these guesses is
approximately correct, and attains the bounded asserted in
\reftheorem{final-randomized-greedy-algorithm}.

\begin{theorem}
  \labeltheorem{final-randomized-greedy-algorithm}
  Let $\pmatchoid = (\groundset,\independents)$ be a $p$-matchoid of
  rank $k$, let $f: 2^{\groundset} \to \nnreals$ a nonnegative
  submodular function over $\groundset$, and let $\epsilon > 0$ be
  fixed. Suppose there exists an algorithm for the offline instance of
  the problem with approximation ratio $\gamma_p$. Then there exists a
  streaming algorithm using total space
  \begin{math}
    O\parof{\frac{k \log k}{\epsilon^2}}
  \end{math} %
  that, given a stream over $\groundset$, returns a set $\bestset \in
  \independents$ such that
  \begin{align*}
    (1 - \epsilon) \OPT \leq \parof{4 p + \frac{1}{\offlineratio}}
    \evof{f(\bestset)}.
  \end{align*}
\end{theorem}


\paragraph{Some notation for the analysis}

\begin{itemize}
\item Let $\finalset$ be the state of $S$ at the end of the stream.
\item Let $\takens$ be the set of all elements to pass through $S$
  during the stream.
\item Let $\finalbuffer$ be the set held by $\buffer$ at the end of
  the stream.
\item Let $\bestset = \argmax_{Z \in \setof{\finalset,S'}} f(Z)$ be
  the set output by \refalgo{Randomized-Streaming-Greedy}.
\end{itemize}
$\finalset$, $\takens$, and $\finalbuffer$ are random sets depending
on the random selection process from $\buffer$. $\bestset$ is a random
variable depending on $\finalbuffer$, $\finalset$, and the offline
algorithm's own internal randomization.

\subsection{Reducing to \refalgo{Streaming-Greedy}}
\labelsection{reduction-to-streaming-greedy}

If we set the buffer limit $K$ to 1, eliminating the role of the
buffer $\buffer$, then \refalgo{Randomized-Streaming-Greedy} reduces
to the deterministic \refalgo{Streaming-Greedy} algorithm from
\refsection{streaming-greedy}. In
\reffigure{randomized-streaming-greedy-reduced}, we refactor
\refalgo[Randomized-Streaming-Greedy]{Randomized-{\allowbreak}Streaming-{\allowbreak}Greedy}
as a buffer placed upstream from a running instance of
\refalgo{Streaming-Greedy}. The buffer only filters and reorders the
stream, and the analysis of \refalgo{Streaming-Greedy} holds with
respect to this scrambled stream.  More precisely, if $r$ denotes the
random bits dictating $\buffer$, then for any fixed $r$, the analysis
of \refsection{streaming-greedy} still applies. We recap the preceding
analysis for \refalgo{Streaming-Greedy} as it applies here.

\begin{figure}[t]
  \centering
  \begin{minipage}{13.2cm}
    \begin{framed}
      \begin{pseudocode}
        \begin{routine}{Randomized-Streaming-Greedy}{$\alpha$,$\beta$} %
          Let $\instance$ be an instance of $\refalgo{Streaming-Greedy}{$\alpha$,$\beta$}$ \\
          Let $S_{\instance}$ refer to the set $S$ maintained by $\instance$. \\
          $\buffer \gets \emptyset$ %
          \commentcode{$\buffer$ is a buffer of size
            $K > 0$}\\
          while (stream is not empty) \\
          \> $e \gets $ next element in the stream \\
          \> if \refalgo{Is-Good}{$S_{\instance}$,$e$} then $\buffer
          \gets \buffer + e$
          \commentcode{buffers the ``good'' elements}\\
          \> else send $e$ downstream to $\instance$
          \commentcode{$\instance$ will reject $e$}\\
          \> if $\sizeof{\buffer} = \bufferlimit$ then \\
          \> \> $e \gets$ an element from $\buffer$ picked uniformly
          at random %
          \commentcode{$e$ is ``good''}\\
          \> \> $\buffer \gets \buffer - e$\\
          \> \> send $e$ downstream to $\instance$ \commentcode{$\instance$ will add $e$ to $S_{\instance}$}\\
          \> \> for all $e' \in \buffer$ such that (not \refalgo{Is-Good}{$S_{\instance}$,$e$}))\\
          \> \> \> $\buffer \gets \buffer - e'$\\
          \> \> \> send $e'$ downstream to $\instance$ %
          \commentcode{$\instance$ will reject $e'$}\\
          \> \> end for \\
          \> end if\\
          end while\\
          $S' \gets \refalgo{Offline}{$\buffer$}$\\
          Return $\argmax_{Z \in \setof{S_{\instance},S'}} f(Z)$
        \end{routine}
      \end{pseudocode}
    \end{framed}
    \caption{\refalgo{Randomized-Streaming-Greedy} rewritten with the
      deterministic portion reduced to \refalgo{Streaming-Greedy}.}
    \labelfigure{randomized-streaming-greedy-reduced}
  \end{minipage}
\end{figure}

\begin{lemma}%
  \labellemma{pointwise-streaming-greedy}
  Let $\alpha, \beta \geq 0$ be fixed parameters. For any $T \in
  \independents$, we have
  \begin{align*}
    f(\takens \cup (T \setminus \finalbuffer)) %
    \leq %
    \frac{(1 + \beta)^2}{\beta} \cdot p \cdot f(\finalset) %
    + k \alpha.
  \end{align*}
  Furthermore, $\sizeof{\takens} \leq \OPT/\alpha$.
\end{lemma}

\subsection{Upper bounding $f(U \cup T)$ by $f(\bestset)$}
\labelsection{online-plus-offline}

To bound $f(U \cup T)$, where $T \in \independents$ is any independent
set, we split $f_{\takens}(T)$ into the portion $f_{\takens}(T
\setminus \buffer)$ that escapes the buffer, and the remainder
$f_{\takens}(T \cap \buffer)$ captured by the buffer. We bound the
former with \reflemma{pointwise-streaming-greedy} and the latter by
guarantees for \refalgo{Offline} to obtain the following.

\begin{lemma}
  \labellemma{nn-submodular-upper-bound-1}
  For any $T \in \independents$, we have
  \begin{align*}
    \evof{f(\takens \cup T)} \leq k \alpha + \parof{\frac{(1 +
        \beta)^2}{\beta} \cdot p + \frac{1}{\offlineratio}}
    \evof{f(\bestset)}
  \end{align*}
\end{lemma}
\begin{proof}
  For ease of exposition, let $r$ denote the random bits that dictate
  the random selections from $B$, and let us subscript variables by
  $r$ to highlight their dependence.  Let $T'_r = T \setminus
  \finalbuffer_r$ and $T''= T \cap \finalbuffer_r$. For any fixed $r$,
  we have,
  \begin{align*}
    f(U_r \cup T) %
    &= %
    f(U_r) + f_{U_r}(T) 
    \leq %
    f(U_r) + f_{U_r}(T'_r) + f_{U_r}(T''_r) %
    &\text{by submodularity,}\\
    &\leq %
    k \alpha + \frac{(1+\beta)^2}{\beta} \cdot p \cdot f(\finalset_r)
    + f_{U_r}(T''_r) %
    &\text{by \reflemma{pointwise-streaming-greedy},}\\
    &\leq %
    k \alpha + \frac{(1 + \beta)^2}{\beta} \cdot p \cdot
    f(\finalset_r) + f(T''_r) %
    &\text{by submodularity,} \\
    &\leq %
    k \alpha + \frac{(1 + \beta)^2}{\beta} \cdot p \cdot
    f(\finalset_r) + \frac{1}{\offlineratio} \evof{f(S'_r)}. %
  \end{align*}
  Here, the expectation surrounding $f(S'_r)$ is generated by the
  offline algorithm $\refalgo{Offline}$, which may be randomized (see,
  for example, \reftable{results}).  Taking expectations of both sides
  over $r$, we have
  \begin{align*}
    \evof{f (\takens \cup T)} %
    &\leq %
    k \alpha + \frac{(1 + \beta)^2}{\beta} \cdot p \cdot
    \evof{f(\finalset)} + \frac{1}{\offlineratio} \evof{f(S')} \\%
    &\leq k \alpha + \parof{\frac{(1 + \beta)^2}{\beta} \cdot p +
      \frac{1}{\offlineratio}} \evof{f(\bestset)},
  \end{align*}
  as desired.
\end{proof}

\subsection{Upper bounding $f(T)$ by $f(U \cup T)$}
\labelsection{competition-vs-competition-and-takens}

The remaining challenge is to bound $\evof{f(U \cup T)}$ from below by
some fraction of $f(T)$. The following technical lemma, used similarly
by Buchbinder \etal, gives us a handle on $\evof{f(U \cup T)}$.
\begin{lemma}[\cite{bfns-smcc-14}] %
  \labellemma{bfns-ev-14}%
  Let $f: \powerset{\groundset} \to \nnreals$ be a nonnegative
  submodular function. Suppose $R$ is a random set according to a
  distribution $\mu$ on $\powerset{\groundset}$ where no element $e
  \in \groundset$ is picked with probability more than $\rho$. Then
  $\evof{f (R)} \geq (1 - \rho) f(\emptyset)$. Moreover, for any set
  $Y \subseteq \groundset$, $\evof{f(R \cup Y)} \geq (1-\rho) f(Y)$.
\end{lemma}

In this case, $\takens$ is a random set, and we want to upper bound
the probability of an element $e \in \groundset$ appearing in
$\takens$. Intuitively, taking $\bufferlimit$ large limits the
probability of an element being selected from the buffer, while by
\reflemma{pointwise-streaming-greedy}, taking $\alpha$ large decreases
the number of elements in $U$.
\begin{lemma}%
  \labellemma{probability-taken}
  For any element $e \in \groundset$,
  \begin{align*}
    \probof{e \in \takens} \leq 1 - \parof{1 -
      \frac{1}{\bufferlimit}}^{\OPT/\alpha}.
  \end{align*}
\end{lemma}
\begin{proof}
  An element is added to $S$ (and therefore $\takens$) if and only if
  it is selected from $\buffer$ when $\sizeof{\buffer}$ reaches
  $\bufferlimit$. By \reflemma{pointwise-streaming-greedy}, we select
  from $\buffer$ at most $\OPT/\alpha$ times, and each selection is
  made uniformly and independently at random from $\bufferlimit$
  elements.
\end{proof}

With this, we apply \reflemma{bfns-ev-14} to give the following.
\begin{lemma}%
  \labellemma{nn-submodular-upper-bound-2} Let $T \in \independents$
  be a fixed independent set. Then
  \begin{align*}
    \evof{f(U \cup T)} \geq \parof{1 -
      \frac{1}{\bufferlimit}}^{\OPT/\alpha} f(T).
  \end{align*}
\end{lemma}
\begin{proof}
  By \reflemma{probability-taken}, for all $e \in \groundset$,
  $\probof{e \in U} \leq \parof{1 - \parof{1 -
      1/\bufferlimit}^{\OPT/\alpha}} \equiv \rho$. The claim then
  follows \reflemma{bfns-ev-14}.
\end{proof}

\subsection{Overall Analysis}
\labelsection{competition-vs-final-set}

Tying together \reflemma{nn-submodular-upper-bound-1} and
\reflemma{nn-submodular-upper-bound-2}, we have the following.
\begin{lemma}%
  \labellemma{nn-submodular-overall-upper-bound} %
  For any $T \in \independents$, we have
  \begin{align*}
    \parof{1 - \frac{\OPT}{\alpha \bufferlimit}} f(T) %
    \leq %
    k \alpha + \parof{\frac{(1 + \beta)^2}{\beta} \cdot p +
      \frac{1}{\offlineratio}} \evof{f(\bestset)}
  \end{align*}
\end{lemma}
\begin{proof}
  Composing \reflemma{nn-submodular-upper-bound-1} and
  \reflemma{nn-submodular-upper-bound-2}, we have,
  \begin{align*}
    \parof{1 - \frac{1}{\bufferlimit}}^{\OPT/\alpha} f(T) %
    \leq %
    k \alpha + \parof{\frac{(1 + \beta)^2}{\beta} \cdot p +
      \frac{1}{\offlineratio}} \evof{f(\bestset)}.
  \end{align*}
  By Bernoulli's inequality,
  \begin{align*}
    \parof{1 - \frac{1}{\bufferlimit}}^{\OPT/\alpha} %
    &\geq %
    1 - \frac{\OPT}{\alpha \bufferlimit},
  \end{align*}
  and the claim follows.
\end{proof}

\subsection{A bound for approximate $\alpha$}
\labelsection{relaxed-competition-vs-final-set} %
\labelsection{relaxed-competition-vs-final-set} %
We would like to fix $\alpha$ as a constant fraction of $\OPT$. For
example, taking $\alpha = \epsilon \OPT / 2 k$, where $\epsilon > 0$,
and plugging into \reflemma{nn-submodular-overall-upper-bound} gives
the cleaner bound,
\begin{align*}
  \parof{1 - \frac{2 k}{\epsilon K}} f(T) %
  \leq %
  \frac{\epsilon}{2} \OPT + \parof{\frac{(1 + \beta)^2}{\beta} \cdot p
    + \frac{1}{\offlineratio}} \evof{f(\bestset)}.
\end{align*}
However, the algorithm does not know $\OPT$, and instead we will try
to estimate $\OPT$ approximately. Let us lay out the bound when
$\alpha$ is within a factor of 2 of $\epsilon \OPT / 2k$.
\begin{lemma}
  Let $\epsilon > 0$ be a fixed parameter.  If $\epsilon \cdot \OPT /
  4 k \leq \alpha \leq \epsilon \cdot \OPT / 2 k$, then
  \begin{align*}
    \parof{1 - \frac{4 k}{\epsilon K}}\OPT %
    \leq %
    \frac{\epsilon}{2} \OPT + \parof{\frac{(1 + \beta)^2}{\beta} \cdot
      p + \frac{1}{\offlineratio}} \evof{f(\bestset)}.
  \end{align*}
  In particular, for $K = 4 k / \epsilon^2$, we have
  \begin{align*}
    (1 - \epsilon) \OPT \leq \parof{\frac{(1 + \beta)^2}{\beta} \cdot
      p + \frac{1}{\offlineratio}} \evof{f(\bestset)}.
  \end{align*}
\end{lemma}

\subsection{Efficiently estimating $\alpha$}
\newcommand{\thresholds}{\mathcal{A}}
\labelsection{estimating-threshold}

Badanidiyuru \etal showed how to ``guess'' $\OPT$ space-efficiently
and in a single pass \cite{bmkk-sso-14}. Let $z = \argmax_{x \in
  \groundset} f(x)$. Clearly, $\OPT \geq f(z)$, and by submodularity
of $f$,
\begin{align*}
  \OPT = \max_{T \in \independents} f(T) \leq \max_{T \in
    \independents} \sum_{t \in T} f(t) \leq \max_{T \in \independents}
  \sizeof{T} \cdot f(z) \leq k \cdot f(z).
\end{align*}
Fix $\epsilon > 0$, and suppose we run a parallel copy of
\refalgo{Randomized-Streaming-Greedy} for each $\alpha$ in
\begin{align*}
  \thresholds(z) = \setof{2^i \where i \in \integers} %
  \cap %
  \bracketsof{\frac{\epsilon}{4k} \cdot f(z), \frac{\epsilon}{2}
    f(z)},
\end{align*}
and at the end of the stream return the best solution among the $\log
(2k)$ copies. For some $\alpha \in \thresholds(z)$, we have
\begin{align*}
  \frac{\epsilon}{4k} \cdot \OPT \leq \alpha \leq \frac{\epsilon}{2k}
  \cdot \OPT,
\end{align*}
where we get the approximation guarantee in
\reflemma{nn-submodular-overall-upper-bound}.

This strategy requires two passes: one to identify $z$, and the second
running $\log(2k)$ copies of
\refalgo[Randomized-Streaming-Greedy]{Randomized-{\allowbreak}Streaming-{\allowbreak}Greedy}
in parallel.  We can reduce the number of passes to 1 by updating $z$
and $\thresholds(z)$ on the fly. Enumerate the stream $e_1,e_2,\dots$,
and for for each $i$, let
\begin{align*}
  z_i = \argmax_{e_j \where j \in [i]} f(e_j)
\end{align*}
be the single element maximizing $f$ among the first $i$ elements seen
thus far. $\thresholds(z_i)$ shifts up over through the stream as
$z_i$ is updated. At each step $i$, we maintain parallel solutions for
each choice of $\alpha \in \thresholds(z_i)$, deleting instances with
$\alpha$ below $\thresholds(z_i)$ and instantiating new instances with
larger values of $\alpha$.

To ensure correctness, it suffices to show that when we instantiate an
instance of
\refalgo[Randomized-Streaming-Greedy]{Randomized-{\allowbreak}Streaming-{\allowbreak}Greedy}
for a new threshold $\alpha$, we haven't skipped over any elements
that we would want to include. Let $\alpha \in \thresholds(z_i) -
\thresholds(z_{i-1})$, i.e.\ $f(z_{i-1}) < \alpha \leq f(z_i)$. If
$f(e_j) \geq \alpha$ for some $j < i$, then
\begin{align*}
  \alpha \leq f(e_j) \leq f(z_{i-1}),
\end{align*}
a contradiction.

\FloatBarrier
\subsection{Simpler algorithm and better bound for cardinality
  constraint}
\labelsection{cardinality}

\begin{figure}
  \centering
  \begin{minipage}{7.15cm}
    \begin{framed}
      \begin{pseudocode}
        \begin{routine}[Randomized-Streaming-Greedy-cardinality]
          {Randomized-Streaming-Greedy}{$\alpha$,$\infty$}
          $B \gets \emptyset$, $S \gets \emptyset$ \\
          while (stream is not empty) \\
          \> $e \gets$ next element in the stream \\
          \> if $\sizeof{S} \leq k$ and $f_S(e) > \alpha$ then \\
          \> \> $B \gets B + e$ \\
          \> if $\sizeof{B} = K$ then \\
          \> \> $e \gets$ uniformly random from $B$  \\
          \> \> $B \gets B - e$, $S \gets S + e$ \\
          \> \> for all $e' \in B$ s.t.\ $f_S(e') \leq \alpha$ \\
          \> \> \> $B \gets B - e'$\\
          \> end if \\
          end while \\
          $S' \gets \refalgo{Offline}{B}$\\
          return $\argmax_{Z \in \setof{S,S'}} f(Z)$
        \end{routine}
      \end{pseudocode}
    \end{framed}
    \vspace{.5em}
    \caption{}
    \labelfigure{randomized-streaming-greedy-cardinality}

  \end{minipage}
\end{figure}

When the $p$-matchoid is simply a cardinality constraint with rank
$k$, we can do better. If we set $\beta = \infty$ in
\refalgo{Randomized-Streaming-Greedy}{$\alpha$,$\beta$}, then the
algorithm will only try to add to $S$ without exchanging while
$\sizeof{S} < k$, effectively halting once we meet the cardinality
constraint $\sizeof{S} = k$. In
\reffigure{randomized-streaming-greedy-cardinality}, we rewrite
\refalgo{Randomized-Streaming-Greedy}{$\alpha$,$\infty$} with the
unnecessary logic removed.

\newcommand{\sizeoffinalset}{|\finalset|}
\begin{lemma} %
  \labellemma{full-set-cardinality-bound} %
  If $\sizeoffinalset = k$, then $f(\finalset) \geq k \alpha$.
\end{lemma}

\begin{lemma}
  \labellemma{unfull-set-cardinality-bound} %
  If $\sizeoffinalset < k$, then for any set $T \subseteq \groundset$,
  \begin{align*}
    f(\finalset \cup T) %
    \leq %
    f(\finalset) + f(T \cap B) + \alpha \sizeof{T}.
  \end{align*}
\end{lemma}
\begin{proof}
  Fix $t \in T \setminus (\finalset \cup B)$, and let $\setbefore{t}$
  be the set held by $\finalset$ when $t$ is processed. Since $t$ is
  rejected, and $S_t^- \subseteq \finalset$, we have
  \begin{align*}
    f_{\finalset}(t) \leq f_{\setbefore{t}}(t) \leq \alpha.
  \end{align*}
  Summed over all $t \in T \setminus (\finalset \cup B)$, we have
  \begin{align*}
    f_{\finalset}(T \setminus B) %
    \leq %
    \sum_{t \in T \setminus B } f_{\finalset}(t) %
    \leq %
    \alpha \sizeof{T}.
  \end{align*}
  Finally, we write
  \begin{align*}
    f(\finalset \cup T) %
    &= f_{\finalset}(T) + f(\finalset) %
    \leq %
    f_{\finalset}(T \setminus B) + f_{\finalset}(T \cap B) +
    f(\finalset) %
    \\
    &\leq %
    f(\finalset) + f(T \cap B) + \alpha \sizeof{T}
  \end{align*} to attain the desired bound.
\end{proof}

\begin{lemma}
  \labellemma{cardinality-constraint-bound}
  For $K = k/\epsilon$, and $\alpha$ such that $(1-\epsilon) \OPT \leq
  (2+e) k \alpha \leq (1+\epsilon) \OPT$, we have
  \begin{align*}
    \evof{f(\bestset)} \geq \frac{1-2\epsilon}{2 + e} \cdot \OPT.
  \end{align*}
\end{lemma}
\begin{proof}
  Let $T \subseteq \groundset$ be an optimal set with $\sizeof{T} = k$
  and $\OPT = f(T)$.

  If $\sizeoffinalset = k$, then the claim follows
  \reflemma{full-set-cardinality-bound}.  Otherwise, $\sizeoffinalset
  < k$ and by \reflemma{unfull-set-cardinality-bound}, we have
  \begin{align*}
    f(\finalset \cup T)\leq f(\finalset) + f(T \cap B) + k \alpha.
  \end{align*}
  By \reflemma{bfns-ev-14}, we also have
  \begin{align*}
    \evof{\finalset \cup T} %
    \geq %
    \parof{1-\frac{1}{K}}^k f(T) %
    \geq \parof{1 - \frac{k}{K}} f(T) %
    \geq %
    (1-\epsilon) f(T).
  \end{align*}
  Finally, by the bound for \refalgo{Offline}, we have,
  \begin{align*}
    f(T \cap B) \leq e \evof{f(S')}
  \end{align*}
  Together, we have
  \begin{align*}
    (1-\epsilon) f(T) %
    \leq %
    \evof{f(\finalset) + e f(S')} + k \alpha %
    \leq %
    \parof{1 + e} \evof{f(\bestset)} + k \alpha.
  \end{align*}
  Solving for $\evof{f(\bestset)}$ and plugging in $f(T) = \OPT$ and
  $k \alpha \leq \frac{1+\epsilon}{2+e} \OPT$, we have
  \begin{align*}
    \evof{f(\bestset)} %
    \geq %
    \frac{1}{1+e} \parof{(1-\epsilon) \OPT - k \alpha} %
    \geq %
    \frac{(1-2\epsilon)}{2+e} \OPT,
  \end{align*}
  as desired.
\end{proof}

The preceding analysis reveals that the appropriate choice for
$\alpha$ is $\OPT / (2+e)k$, where $\OPT = \maxof{f(T) \where
  \sizeof{T} \leq k}$ is the maximum value attainable by a set of $k$
elements, and that a sufficiently large choice for $K$ is
$k/\epsilon$. As in \refsection{estimating-threshold}, we can
efficiently approximate $\alpha$ by guessing $\alpha$ in increasing
powers of $(1+\epsilon)$, maintaining at most $O(\log_{1+\epsilon} k)
= O(\epsilon^{-1} \log k)$ instances of
\refalgo[Randomized-Streaming-Greedy-cardinality]{Randomized-Streaming-Greedy}{$\alpha$,$\infty$}
at any instant. The resulting bound is stronger than previously
derived for a 1-matchoid.

\begin{theorem}
  Let $f: 2^{\groundset} \to \nnreals$ be a nonnegative submodular
  function over a ground set $\groundset$, and let $\eps > 0$ be
  fixed. Then there exists a streaming algorithm using total space
  \begin{math}
    O\parof{\frac{k \log k}{\epsilon^2}}
  \end{math} %
  that, given a stream over $\groundset$, returns a set $\bestset$
  such that $|\bestset| \leq k$ and
  \begin{math}
    f(\bestset) %
    \geq %
    \frac{1-\epsilon}{2 + e} \cdot \OPT %
  \end{math}, %
  where $\OPT = \maxof{f(T) \where \sizeof{T} \leq k}$ is %
  the maximum value attainable by a set of $k$ elements.
\end{theorem}


\FloatBarrier
\section{A Deterministic Algorithm via Iterated Greedy}
\labelsection{iterated-streaming-greedy}
\labelsection{iterated-streaming-greedy}
\begin{figure}[t]
  \centering
  \begin{minipage}{7.5cm}
    \begin{framed}
      \begin{pseudocode}
        \begin{routine}{Iterated-Streaming-Greedy}{$\alpha$,$\beta$,$\groundset$}
          \commentcode{run \refalgo{Streaming-Greedy} over $\groundset$}\\
          $(S_1,U_1) \gets
          \refalgo{Streaming-Greedy}{$\alpha$,$\beta$,$\groundset$}$
          \\
          \commentcode{$U_1$ denotes the set $U$ in
            \refsection{streaming-greedy}.}
          \\
          $S_2 \gets
          \refalgo{Streaming-Greedy}{$0$,$\beta$,$\groundset
            \setminus U_1$}$\\
          $S_3 \gets \refalgo{Offline}{$U_1$}$\\
          return $\argmax_{\bestset \in \setof{S_1,S_2,S_3}}
          f(\bestset)$
        \end{routine}
      \end{pseudocode}
    \end{framed}
  \end{minipage}
\end{figure}

Gupta \etal gave a framework that takes an offline algorithm for
maximizing a monotone submodular functions and, by running the
algorithm as a black box multiple times over different groundsets,
produces an algorithm for the nonnegative case \cite{grst-10}. Here we
adapt the framework to the streaming setting, employing
\refalgo{Streaming-Greedy} as the blackbox for the monotone case.


We first present \refalgo{Iterated-Streaming-Greedy} as an algorithm
making two passes over $\groundset$. In the first pass we run
\refalgo{Streaming-Greedy}{$\alpha$,$\beta$} over $\groundset$ as
usual.  Let $S_1$ denote the set output, and $U_1$ the set of all
elements added to $S_1$ at any intermediate point of the algorithm, as
per \refsection{streaming-greedy}. In the second pass, we run
\refalgo{Streaming-Greedy}{$0$,$\beta$} over the set $\groundset
\setminus U_1$ of elements that were immediately rejected in the first
pass to produce another independent set $S_2$. Lastly, we run our
choice of offline algorithm over $U_1$ to produce a third independent
set $S_3 \in \independents$. At the end, return the best set among
$S_1$, $S_2$, and $S_3$.

If we pipeline the two instances of \refalgo{Streaming-Greedy}, then
\refalgo{Iterated-Streaming-Greedy} becomes a true streaming algorithm
with only one pass over $\groundset$. When the first instance rejects
an element $e$ outright, $e$ is sent downstream to the second instance
of \refalgo{Streaming-Greedy}. Note that the running time of
\refalgo{Offline} depends on the size of its input $U_1$, which by
\reflemma{size-of-tokens} is at most $\OPT/\alpha$.

\begin{lemma}
  Let $T \in \independents$ be any independent set. Then
  \begin{align*}
    \parof{2 \cdot \frac{(1+\beta)^2}{\beta} \cdot p +
      \frac{1}{\gamma_p}}
    \evof{f(\bestset)} \geq f(T) - k \alpha
  \end{align*}
  Furthermore, if \refalgo{Offline} is a deterministic algorithm, then
  \refalgo{Iterated-Streaming-Greedy} is deterministic and the above
  holds without taking expectations.
\end{lemma}
\begin{proof}
  By submodularity, we have
  \begin{align}
    f(U_1 \cup T) + f(U_2 \cup (T \setminus U_1)) %
    \geq %
    f(U_1 \cup S_U) + f(T \setminus U_1) %
    &\labelequation{ig1}
  \end{align}
  and
  \begin{align}
    f(T \setminus U_1) + f(U_1 \cap T) \geq f(T) + f(\emptyset).
    \labelequation{ig2}
  \end{align}
  By nonnegative of $f$ and \refequation[equations]{ig1}
  and\refequation[]{ig2}, we have,
  \begin{align*}
    f(T) %
    &\leq %
    f(T) + f(\emptyset) + f(U_1 \cup S_2) %
    \leq %
    f(U_1 \cup T) + f(U_2 \cup (T \setminus U_1)) + f(U_1 \cap T). %
  \end{align*}
  By \refcorollary{monotone-bounds}, we have $f(U_1 \cup T) \leq
  \frac{(\beta + 1)^2}{\beta} \cdot p \cdot f(S_1) + k \alpha$ and
  $f(U_2 \cup (T \setminus U_1)) \leq \frac{(\beta + 1)^2}{\beta}
  \cdot p \cdot f(S_2)$, and $f(U_1 \cap T) \leq \frac{1}{\gamma_p}
  \cdot f(S_3)$ by assumption. Plugging into the above, and noting
  that $f(S_1),f(S_2),f(S_3) \leq f(\bestset)$ gives the bounds we
  seek.
\end{proof}
\begin{corollary}
  Let $\epsilon > 0$ be given. If $\alpha \leq \epsilon \OPT / k$,
  then
  \begin{align*}
    \parof{2\frac{(1+\beta)^2}{\beta} \cdot p + \frac{1}{\gamma_p}}
    \evof{f(\bestset)} %
    \geq%
    (1-\epsilon) \OPT,
  \end{align*}
  and the inequality holds without taking expectations if
  \refalgo{Offline} is a deterministic algorithm.
\end{corollary}

The appropriate value of $\alpha$ is guessed efficiently exactly as
described in \refsection{estimating-threshold}. Here, if
$\sizeof{U_1}$ grows too large in an instance of
\refalgo{Iterated-Streaming-Greedy}{$\alpha$,$\beta$} for some fixed
$\alpha$, then $\alpha$ must be too small and we can terminate the
instance immediately.

\begin{theorem}
  Let $\defpmatchoid$ be a $p$-matchoid of rank $k$, let $f:
  \powerset{\groundset} \to \nnreals$ be a nonnegative submodular
  function over $\groundset$, let $\epsilon > 0$ be fixed. Suppose
  there exists an offline algorithm for finding the largest value
  independent set in $p$-matchoid with approximation ratio $\gamma_p$.
  Then there exists a streaming algorithm using total space
  \begin{math}
    O\prac{k \log k}{\epsilon}
  \end{math} %
  that, given a stream of $\groundset$, returns a set $\bestset \in
  \independents$ such that
  \begin{align*}
    \parof{8p + \frac{1}{\gamma_p}} \evof{f(\bestset)} \geq
    (1-\epsilon) \OPT.
  \end{align*}
  If the offline algorithm is deterministic, then the claimed
  algorithm is deterministic and the above bound holds without
  expectation.
\end{theorem}


\bibliographystyle{alpha}%
\bibliography{submodular_maximization_in_streams}

\appendix

\section{Exchange lemmas for matroids}

\begin{lemma}
  \labellemma{transitivity-span}
  Let $\defmatroid$ be a matroid, let $S,T \subseteq \groundset$ be
  two subsets, and let $x,y\in \groundset$ be two elements.  If $S$
  spans $x$, and $T + x$ spans $y$, then $S \cup T$ spans $y$.
\end{lemma}
\begin{proof}
  It suffices to assume that $S$ and $T$ are independent sets.

  Extend $S$ to a base $\baseA$ in $S \cup T$. Since $\baseA$ extends
  $S$, $\baseA$ spans $x$, and $\baseA$ is a base in $(S \cup T) +
  x$. Since $(S \cup T) + x$ spans $y$, $\baseA$ spans $y$.
\end{proof}

Let $\graph$ be a directed graph. For $v \in \vertices(\graph)$, let
$\outneighborhoodof{v} = \setof{w \suchthat \arc{v}{w} \in
  \edges\parof{\graph}}$ denote the set of outgoing neighbors of $v$,
and $\inneighborhoodof{v} = \setof{u \suchthat \arc{u}{v} \in
  \edges(\graph)}$ the set of incoming neighbors of $v$.

The following lemma is implicit in Badanidiyuru \cite{abv-11}.
\begin{lemma}%
  \labellemma{graphical-span-matching}
  Let $\defmatroid$ be a matroid, and $\graph$ a directed acyclic
  graph over $\groundset$ such that for every non-sink vertex $e \in
  \groundset$, the outgoing neighbors $\outneighborhoodof{e}$ of $e$
  {span} $e$. Let $\defindset$ be an independent set such that no path
  in $\graph$ goes from one element in $\indset$ to another. Then
  there exists an injection from $\indset$ to sink vertices in
  $\graph$ such that each $e \in \indset$ maps into an element
  reachable from $e$.
\end{lemma}
\begin{proof}
  Restricting our attention to elements reachable from $\indset$ in
  $\graph$, let us assume that the elements of $\indset$ are sources
  (i.e., have no incoming neighbors) in $\graph$. Let us call an
  element $e \in \groundset$ an ``internal'' element if its is neither
  a sink nor a source in $\graph$.

  We prove by induction on the number of internal vertices reachable
  from $\indset$. In the base case, the outgoing neighbors of each $i
  \in \indset$ are all sinks, and $\graph$ is bipartite. For any
  subset $\indsetB \subseteq \indsetA$, $\outneighborhoodof{\indsetB}
  = \bigcup_{i \in \indsetB} \outneighborhoodof{i}$ spans $\indsetB$.
  $\indsetB$ is independent, so we have $\sizeof{\indsetB} \leq
  \sizeof{\outneighborhoodof{\indsetB}}$. Thus, by Hall's matching
  theorem, there exists an injection $\indset \into
  \outneighborhoodof{\indset}$ such that each $i \in \indset$ maps
  into $\outneighborhoodof{i}$.

  In the general case, let $e \in \groundset \setminus \indset$ be an
  internal vertex in $\graph$. Consider the graph $\graphB$ removing
  $e$ and preserving all paths through $e$, defined by,
  \begin{align*}
    \vertices(\graphB) &= \vertices(\graph) - e = \groundset - e,\\
    \edges(\graphB) &= %
    \parof{%
      \edges (\graph) %
      \setminus %
      \setof{\arc{a}{b} \where a = e \text{ or } e = a}%
    } %
    \cup %
    \setof{\arc{a}{b} \where \arc{a}{e}, \arc{e}{b} \in
      \edges(\graph)}
  \end{align*}
  $\graphB$ has one less internal vertex than $\graph$, the same sink
  vertices as $\graphA$, and a vertex $a \in \vertices(\graphB)$ is
  reachable from $i \in \indset$ in $\graphB$ iff it is reachable from
  $i$ in $\graphA$. For any vertex $a \in
  \inneighborhoodof[\graphA]{e}$ that had an outgoing arc into $e$, we
  have
  \begin{align*}
    \outneighborhoodof[\graphB]{a} = (\outneighborhoodof[\graphA]{a} -
    e) \cup \outneighborhoodof[\graphA]{e},
  \end{align*}
  which spans $a$ by \reflemma{transitivity-span}. Since any other
  vertices has the same outgoing arcs, we conclude that
  $\outneighborhoodof[\graphB]{a}$ spans $a$ for any non-sink vertex
  of $\graphB$.

  By induction, there exists an injection from $\indset$ into the
  sinks of $\graphB$ such that every $i \in \indset$ is mapped to a
  sink vertex reachable from $i$ in $\graphB$. By construction, these
  vertices are also reachable sinks in $\graphA$, as claimed.
\end{proof}

\end{document}